\documentclass[10pt]{IEEEtran}

\usepackage{cite}
\usepackage{graphicx}
\usepackage{amsmath}
\usepackage{enumerate}
\usepackage{amsthm}
\usepackage{mathtools}
\usepackage[linesnumbered,lined,boxed,commentsnumbered,ruled]{algorithm2e}
\usepackage[noend]{algpseudocode}

\newtheorem{prop}{Proposition}
\newtheorem{thm}{Theorem}
\newtheorem{cor}{Corollary}

\theoremstyle{definition}

\makeatletter
\def\BState{\State\hskip-\ALG@thistlm}
\makeatother
\usepackage{color}
\usepackage{amssymb}

\DeclarePairedDelimiter\floor{\lfloor}{\rfloor}

\IEEEoverridecommandlockouts
\begin{document}
	
	\title{Local Relay Selection in Presence of Dynamic Obstacles in Millimeter Wave D2D Communication
		\thanks{\noindent Durgesh Singh and Sasthi C. Ghosh are with the Advanced Computing \&  Microelectronics Unit, Indian Statistical Institute, Kolkata 700108, India. Email: durgesh.ccet@gmail.com, sasthi@isical.ac.in.}
		\thanks{\noindent Arpan Chattopadhyay is with the Department of Electrical Engineering, Indian Institute of Technology Delhi. Email: arpanc@ee.iitd.ac.in.}
		\thanks{\noindent This work was supported by the faculty seed grant and professional development allowance (PDA) of IIT Delhi.}
	}
	\author{\IEEEauthorblockN{
		Durgesh Singh, Arpan Chattopadhyay \& Sasthi C. Ghosh}
	}
	\maketitle

	\begin{abstract}	
		Blockage due to obstacles in millimeter wave (mmWave) device to device (D2D) communication  is a prominent problem due to their severe penetration losses. Potential user equipments (UEs) in vicinity of the source UE must be explored in order to select a new relay when the current link gets blocked. However, dynamic obstacles are not known in advance and thus may cause unpredictable fluctuations to D2D channel quality causing newly selected relay link also to be susceptible to blockage. This might cause frequent relay  switching leading to call drops and high energy consumption. We have proposed the idea of reducing frequency in relay exploration and switching and thus average end-to-end delay (in seconds) at the expense of additional exploration time units (few milliseconds) during beam alignment. We seek to learn the uncertainty in D2D link qualities by modeling the problem as finite horizon partially observable Markov decision process (POMDP) framework locally at each UE. We have derived an optimal threshold policy which maps the state to set of actions. We then give a simplified and easy to implement stationary threshold policy which counts the number of successive acknowledgment successes/failures for making decisions of selecting or not selecting a given relay locally. Through extensive simulation, we validate our theoretical findings and demonstrate that our approach captures the trade-off between average exploration time and average end-to-end (E2E) delay in presence of dynamic obstacles.
	\end{abstract}
		\begin{IEEEkeywords}
		Relay selection, Millimeter wave \texttt{D2D} communication,  Exploration, Dynamic Obstacles, POMDP.
		\end{IEEEkeywords}
	
	\IEEEpeerreviewmaketitle
	
	\section {Introduction}
	Device to device (D2D)  communication enables proximity devices or user equipments (UEs) to directly communicate with one another bypassing the base station (BS) \cite{8618386}. High available bandwidth and short range transmission of millimeter wave (mmWave) is a lucrative choice  for D2D communication \cite{7010536,8014297,8472783}. Although very high propagation losses of mmWave are compensated by transmitting directional beams using  antennas arrays, severe penetration loss of mmWave make it susceptible to very high blockage from various obstacles \cite{7010536,8047278,8941039}. Obstacles  may block the mmWave signal completely,  thus needing an almost line of sight (LOS) path. Additionally, the presence of dynamic obstacles might deteriorate D2D channel condition rapidly \& abruptly causing unprecedented link breakage and hence  packet loss, delay and high energy consumption. 
	
	In case of blockage by dynamic obstacle, a new relay is chosen by exploring new relay links through beam-forming \cite{7744807,8809571} which are in the source UE's vicinity. This is a directional search of new relay and has a considerable delay. A detailed techniques for beam-alignment and beam-management under various  scenarios can be found in \cite{8458146}. The new relaying UE must be chosen carefully  during exploration time (ranging from few microseconds to $10~ms$ \cite{7744807}  for short range), because it may get blocked during data transmission time (100-1000 times higher than exploration time) due to presence of dynamic obstacles, even when the source and relay beams are perfectly aligned to achieve highest data rate as shown in figure \ref{fig_problem}. This might lead to frequent relay exploration and switching which causes increased delay, energy consumption and probably outage leading to call drops. In addition, exchange of channel state information with the BS might cause extra delay \cite{8302856} since the D2D channel is not directly visible to the BS. Hence the decision to explore and select a relay must be made locally.
    \begin{figure}[h!]
		\centering
		\includegraphics[width=0.30\textwidth]{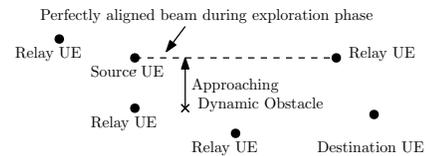}
		\caption{Perfectly aligned beam of source-destination during exploration susceptible for blockage due to dynamic obstacles.}
		\label{fig_problem}
	\end{figure}
	
	Most of the works \cite{6932503,p8_7510705,p14_7450161,8292574} deal with static obstacles. However, dynamic obstacles are unpredictable which must be captured and their effects taken into account. The authors in \cite{8941039} showed a significant drop in data rate when a pedestrian act as a blockage. The authors in \cite{8457255,8990741} used radar to capture the movement of UEs and obstacles. Vision cameras and machine learning (ML) techniques were used by the authors in \cite{8472783,8930580,8941039,9129369} for tracking obstacle's spatio-temporal behavior. However all these require either expensive hardware or high processing time and energy usage. These might be apt at the BS but not at the UE which needs a local solution in an online fashion. The ML based solution in \cite{8930580} might need re-training when there is change in the environment. Moreover, it requires that the link breakage event must follow some well defined pattern with some known distribution. However, due to the presence of obstacles, variations in link quality behavior is abrupt. So these factors may not be perfectly modeled as also argued in \cite{7959158}, thus requiring an online approach. The authors in \cite{8472783} mentioned an online ML technique at the BS which learns about the dynamics of the environment for enabling beam training  to prevent blockages. However in this work we aim to learn the dynamics of the environment at the UE locally in a timely manner. To account for the dynamic obstacles locally, partially observable Markov decision process (POMDP) \cite{7080987,7996366,dp_book}  can be used in modeling the variations in unobserved D2D links. The authors in \cite{9148816}  utilized POMDP to derive an optimal policy which tells the UE locally when to go for exploration after suspending communication on current relay link after successive packet losses. Whereas, in contrast, this work deals with selecting an appropriate relaying UE during the exploration phase locally. The probabilistic model and cost structure are also accordingly different in this work.

	In our work, we have investigated the idea of reducing frequency in relay switching and thus average end-to-end (E2E) delay at the expense of additional exploration time during beam alignment. We have modeled the problem as a finite horizon POMDP. The states are the D2D relay link qualities which are not observable at the current time instant. It can only be observed after receiving the acknowledgements (ACKs) of the probe packets which are sent in order to access a link.  \textit{Even the ACKs can get lost due to presence of dynamic obstacles}.  Information about dynamic obstacles are \textit{not} known at BS a priori and it can only be learned through ACKs of probe packets.  The goal is to take decision of whether selecting or not selecting a relay which minimizes packet loss and in-turn delay by sending additional exploration probe packets to learn the channel quality. Optimal threshold policies have been derived which maps the  belief to a set of actions.  By exploiting the derived policy structure, we have obtained a stationary policy which tells the UE during exploration that after how many successive ACK successes  or ACK  failures to take the decision  of whether selecting or not selecting the relay respectively. Theoretical analysis is validated through extensive simulation. Our major contributions in this paper are summarized as follows:
	\begin{enumerate}
		\item We modeled relay selection problem during the exploration time locally considering the presence of dynamic obstacles as a POMDP. This model can be applied to a scenario without explicitly knowing the dynamic obstacle's distribution.
		\item The model learns the channel quality in an online fashion using ACKs which can also get lost. The model does not require to undergo a training  phase and hence suitable when there is a  change in the environment.  
		\item  We showed that the optimal policy is threshold type policy. This is a non-trivial result that required proof of several interesting intermediate results. Our optimal policy can be implemented locally at each UE, thereby facilitating distributed implementation.
		\item The threshold policy for the problem is further reduced to counting the number of successive ACK successes or ACK failures, which is simple and easy to implement.
	\end{enumerate}


	\section{System Model} \label{system_model}
	We are considering the device-tier of 5G D2D architecture mentioned in \cite{tehrani2014device}, where devices or UEs can communicate among themselves with or without the help from BS. The service region is discretized into various \textit{zones} or \textit{grids} as shown in figure \ref{fig_grid_time}(a) with one BS. Each zone may have many UEs and is assumed to have at least one D2D device which is ready to take part in D2D communication as a relay or source/destination node. 
	A zone $i$ containing the source UE may form connection to a UE of another zone $j\in \mathbb{U}^i$ , where $\mathbb{U}^i$ is the \textit{viable} relay zones of the zone $i$ which is given by the BS. A viable relaying zone of zone $i$ is one which is nearer to the zone containing the destination UE and is in the communication range of the zone $i$. When the source UE in the fixed zone $i$ forms a connection with another UE of zone $j\in \mathbb{U}^i$, then the link formed between zones $i$ and $j$ is termed as link $j$.	Link is formed between UEs of two zones when they are in communication range of each other and the received signal strength is sufficient for the required data rate.
	Each UE can communicate with one another on mmWave channels using directional antennas. The received signal strength ($Q_{ij}$) on zone $j$ from zone $i$  is modeled as \cite{p14_7450161}:
	\begin{equation}
	Q_{ij}=\mu \cdot P \cdot G_t\cdot G_r \cdot PL_{ij}
	\end{equation}
	where, $\mu$ is the shadowing random variable, $P$ is the transmit power of the UE, $G_t$ \& $G_r$ are transmit and receive beam-forming gains respectively. $PL_{ij}$ is the distance dependent path loss function between zone $i$ and zone $j$.

	Time is discretized  as $(nN+l)\delta$ as shown in figure \ref{fig_grid_time}(b), where $n$ belongs to set of nonnegative integers, $l$ takes integer values in $[0,N-1]$, $\delta$ is the smaller discretized time slot when the UEs transmit packets locally. It is assumed that $\delta$ (for each $l\in[0,N-1]$) is large enough to send one packet of size $L$ bytes. Here, $N$ is the number of time slots (of $\delta$ duration) between two consecutive global decisions by the BS. Global decision by BS is made at time when $nN+l$ is divisible by $N$. At this time instant BS takes the channel state information from all UEs in the service region and gives the decision of best relaying UE of a given zone for a given source UE. Hence, in between two consecutive time instants when BS can make global decision, a UE can send at-most $N$ packets of size $L$ to another UE. Note that at time $l=0$, the UE chooses the relay link suggested by the BS and at time  $l\in \{1,2,\cdots,N-1\}$, UEs do not get channel state information from the BS.  At global time instants, BS sends two types of information to UEs, i) the best relaying UE for a given source UE and ii) viable relaying zones $\mathbb{U}^i$ for given source zone $i$, hence the zone $i$ may choose  an appropriate zone for relaying data from the set $\mathbb{U}^i$ by undergoing exploration.
	\begin{figure}[h!]
		\centering
		\includegraphics[width=0.44\textwidth]{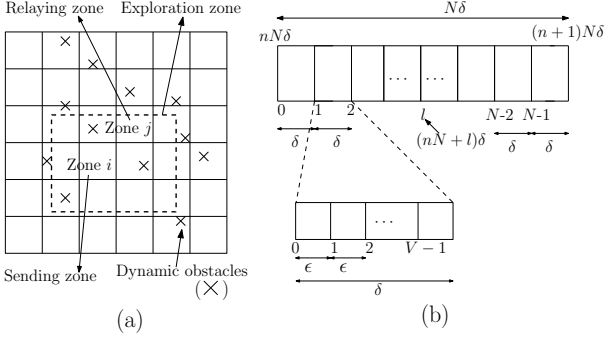}
		\caption{(a) Service region divided into zones along with dynamic obstacles. (b) Discretized time slots with exploration time unit.}
		\label{fig_grid_time}
	\end{figure}

	\textbf{Exploration}: When the current link quality is not good enough then source UE locally explores for an appropriate relaying UE from given viable set. One of the basic procedure in the exploration phase involves sequentially searching space in all directions to align transmitter and receiver beams (beam alignment). The quality of link is piggy-back to the sending UE. This overall process of beam alignment is termed as the exploration phase. Any of the state of art approaches mentioned in \cite{8458146} can be applied for the exploration phase. We have considered an abstraction of this phase and thus it is considered as a black box. It is assumed that the beams are perfectly aligned after the exploration phase and thus we focus on the effects of dynamic obstacle's blockage on a given D2D link. This exploration by source UE in zone $i$ is done for UEs belonging only in the set $\mathbb{U}^i$ to find out the best relaying zone for that time instant. Note that the UE is using directional mmWave antennas for exploring the neighbors and hence explorations cause some significant delay. Each exploration time duration is denoted by $\epsilon < \delta$ as shown in figure \ref{fig_grid_time}(b). In a given duration $\delta$, we can perform a number of explorations. Each of the time instants during exploration time of duration $\epsilon$  is denoted by $m\in\{0,1,\cdots,V-1\}$, where $V=\floor*{\delta/\epsilon}$.
	The overall exploration time is bounded by the maximum value $M$ which can be more than $V$. It is assumed that once exploration is complete,  switching takes negligible time.

	There are static and dynamic obstacles in the service region. There is \textit{no} external facility like radars or vision cameras available at BS to track them. The behavior of dynamic obstacles  are not known a priori and need to be learned from the received ACKs of probe packets during exploration phase in an online fashion. Since mmWaves are highly susceptible to obstacles and suffer from severe penetration losses, we assume that even a single moving or static obstacle may break an already established D2D link and can cause packet loss. 

	\section{Problem Formulation as POMDP} \label{problem_formulation}
	Zone $i$ containing source UE forms a link with an UE (relay or destination) of some other zone $j\in \mathbb{U}^i$. Global decision for the best relay is given by the BS at the time instant $nN\delta$ to relay data packet till $(n+1)N\delta$ time instant. The quality of link given by the BS may deteriorate due to presence of dynamic obstacles resulting in link outage and a new relay link needs to be explored from the viable set $\mathbb{U}^i$ locally. A newly found relay link might provide very good data rate initially during exploration phase, but it may gradually deteriorate in future data transmission time because of dynamic obstacles. This problem is subtle due to the fact that exploration time is very small (ranging from few microseconds to $10~ms$ \cite{7744807}) in contrast to the data transmission time which is in order of seconds (usually 100-1000 times more than search time). The main idea is to reduce the E2E delay on the cost of sending additional probe packets till decision of whether selecting or not selecting a given relay for data transmission considering dynamic obstacles can be learned. This process is limited upto  maximum $M$ exploration time units where each exploration time duration is of $\epsilon$ time. If the given relay is not selected within this time then the exploration process again starts for another possible relaying UE in a zone belonging to $\mathbb{U}^i$.
	
	We model this problem of relay selection during  exploration as that of POMDP considering uncertainty in relay link quality. We will define the POMDP as follows: For all the possible links $j\in \mathbb{U}^i$, the state is written as $y_{m}^j \in\{0,1\}$ for exploration time instant $m$. Values $1$ and $0$ signify if relay link $j$ is in  good or bad state respectively. Good and bad notion denotes that whether link will be formed successfully or not respectively. State represented by $G$ and  $\overline{G}$ respectively denote the good and bad state as shown in figure \ref{fig_prob_2}. The action set is defined as \{do not select relay link,  select relay link, decision cannot be made thus continue sending probe packets\}.  The first action signifies that the current link is bad, hence stop exploring it and start a new exploration. The second action signifies that the current link being explored is of good quality, hence choose it for data transmission. The third action signifies that the decision of relay selection cannot be made and exploration should continue for another $\epsilon$ unit of time to learn the relay link quality. The action is denoted by $a_m^j$ for link $j$ at time instant $m$. The link quality is observed by the ACKs of the received probe packet signal. The  ACK test for link $j$ at time $m$ is denoted as $w_{m}^j\in\{0,1\}$ for a given state $y_m^j$ and action $a_m^j$, may also be uncertain due to the unpredictable channel condition.  In figure \ref{fig_prob_2}, $A$ and $\overline{A}$  denotes whether ACKs are received successfully or not respectively.
	  
	Figure \ref{fig_prob_2} represents the probabilistic structure of the problem as shown below.
	\[P(w_m^j=1|y_{m}^j=1)=k; P(w_m^j=0|y_{m}^j=1)=1-k\]
	\[P(w_m^j=1|y_{m}^j=0)=g; P(w_m^j=0|y_{m}^j=0)=1-g\]
	However, note that the ACK piggyback in current time instant will show quality of the channel for probe packet sent in previous time instant. Hence a transition probability of $g$ is introduced above. 
	The probabilistic structure assumed for the system state transition is given below. We have assumed that $q>s$ and $k>g$ which signify respectively that the probability of a link becoming good from previous good state is higher than that of previous bad state and probability of successful ACK from good state is higher than that of bad state.
	\[P(y_{m+1}^j=1|y_{m}^j=1)=q; 	P(y_{m+1}^j=0|y_{m}^j=1)=1-q\]
	\[P(y_{m+1}^j=1|y_{m}^j=0)=s; 	P(y_{m+1}^j=0|y_{m}^j=0)=1-s\]
	\begin{figure}[h!]
		\centering
		\includegraphics[width=0.25\textwidth]{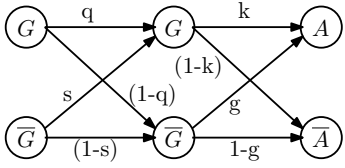}
		\caption{Probabilistic structure of the exploration problem at a UE locally.}
		\label{fig_prob_2}
	\end{figure}
	For a given relaying zone $j$, let us define the information vector available locally to the zone $i$ till time instant $m$ as $H_m^j=(w_0^j,w_1^j,\cdots,w_m^j)$. The sufficient statistics or belief \cite{dp_book} (chapter 5) for this problem is defined locally for relaying zone $j$ as:
	\begin{equation}
	r_m^j=P(y_m^j=1|H_m^j)
	\end{equation}
	This equation gives the probability that a relaying link/zone is in good state given the previous history of information. The estimator function for the local system can be defined as:
	\begin{equation}
	r_{m+1}^j=\Phi'(r_m^j,w_{m+1}^j).
	\end{equation}
	Using Baye's rule we get,
	\begin{equation} \label{update_rule_2}
	r_{m+1}^j = \begin{cases}
	\frac{(qr_m^j+(1-r_m^j)s)k}{H}, & \text{if } w_{m+1}^j=1\\
	\frac{(qr_m^j+(1-r_m^j)s)(1-k)}{J} , & \text{if }w_{m+1}^j=0
	\end{cases}	
	\end{equation}
	where,
	 \[H=(qr_m^j+(1-r_m^j)s)k+(r_m^j(1-q)+(1-r_m^j)(1-s))g\]
	  \[J=(qr_m^j+(1-r_m^j)s)(1-k)+(r_m^j(1-q)+(1-r_m^j)(1-s))(1-g)\]
	
	\textbf{Cost structure}: The cost for stopping the exploration for both cases when link cannot be and can be formed is $0$. However, if the link which seems to be good and selected, may become bad in upcoming time instants and thus causes packet loss.  So it will incur $D_1$ cost to compensate for the packet loss. Similarly, the link which was bad and not selected could have been a good relay link in upcoming time instants. Then it will incur some cost $D_2$ in order to compensate for further exploration.  If decision of relay selection cannot be made in current time instant, then probe packets are sent for another $\epsilon$ time unit to continue the exploration.  The cost incurred here is $c_\epsilon$. This will go on till time $M$ which is the upper bound on the exploration time for a single link. The objective is to derive a decision criterion, whether to select or not select the given relay link, or if this decision cannot be made then continue learning the current relay link quality. The expected cost is formulated as a dynamic program. At the end of the $(M-1)th$ period, the expected cost is: 
	\begin{equation}
	K_{M-1}^j(r)= \min\{rD_1,(1-r)D_2\} \label{eq_j11}
	\end{equation}
	where $r$ is a variable denoting belief. For the time instant $m=M-2$, we have,
	\begin{equation}
	K_{M-2}^j(r)=\min\{ rD_1,(1-r)D_2,c_\epsilon+\mathbb{E}[K_{M-1}^j(r)] \label{eq_j12}
	\end{equation}
	where $\mathbb{E}[\cdot]$ is the mathematical expectation over the observations. Here the first term in minimization expression denotes the expected cost ($rD_1+(1-r)0$) incurred when the learning of the current link being explored  is stopped since it was in the bad state. The second term denotes the expected cost for selecting the current relay link being explored. The third cost is the expected cost for continuing exploring the current relay link for another round of $\epsilon$ unit of time. We can write a general expression as:
	
	\scriptsize	
	\begin{equation}
	K_{m}^j(r)=\min\{rD_1,(1-r) D_2, c_\epsilon+\mathbb{E}[K_{m+1}^j(\Phi'(r,w))]\} \label{eq_j13}
	\end{equation}\normalsize
	where $w\in \{0,1\}$ is a variable denoting ACKs failure/success. Note that the notation $w_m^j$ defined earlier denotes the ACK failure/success of link $j$ at time instant $m$, whereas $w$ is used to denote a general variable for ACK failure/success. We will derive the optimal policy for this problem in section \ref{derivation_policy_2}. 
	
	\section{Derivation of the Optimal Policy in  Exploration Phase} \label{derivation_policy_2}
	
	\subsection{Properties of $K_m^j(r)$}
	The general expression for time instant $m$ as mentioned in equation \eqref{eq_j13}, can be written equivalently as:
	\begin{equation} 
	K_{m}^j(r)=\min\{rD_1,(1-r) D_2,E_m^j(r)\} \label{eq_dp_2}
	\end{equation}
	where, 
	\begin{multline} \label{eq_j14}
	E_m^j(r)=c_\epsilon+P(w=1|r) K_{m+1}^j(\Phi'(r,1)) +\\ P(w=0|r)  K_{m+1}^j(\Phi'(r,0))
	\end{multline}	
	For notation simplicity we will now remove the superscript $j$ from each of the respective notations, e.g., we will write $E_m^j()$ as $E_m()$. Hence $E_m^j(r)$ can now be denoted as $E_m(r)$. 
	Now equation \eqref{eq_j14} can be further reduced as:
	
	\scriptsize
	\begin{multline} \label{eq_j15}
	E_m(r)=c_\epsilon +\\ \big((qr+(1-r)s)k+(r(1-q)+(1-r)(1-s))g\big)K_{m+1}(\Phi'(r,1)) +\\ \big(1-\big((qr+(1-r)s)k+(r(1-q)+(1-r)(1-s))g\big)\big)K_{m+1}(\Phi'(r,0))
	\end{multline}
	\normalsize
	The expected cost for the base case at the end of  $(M-1)th$ period is mentioned in equation \eqref{eq_j11}. 
	We now show some of the properties of $E_m(r)$.
	\begin{prop} \label{prop_3}
		$E_m(r)$ is piece-wise linear and concave in $r$.
	\end{prop} 	
	\begin{proof}
		Proof given in Appendix \ref{A6}.
	\end{proof}
	
	\begin{prop}\label{prop_4}
		$\forall r \in [0,1]$, 	\[E_{m-1}(r) \le E_m(r) \le E_{m+1}(r)\]
	\end{prop}
	
	\begin{proof}
		Proof given in Appendix \ref{A7}.
	\end{proof}
	
	\subsection{Policy Structure}
	The structure of an optimal policy for our POMDP problem is provided in the following theorem.
	\begin{thm}\label{thm_2}
		The optimal policy for exploration in POMDP problem is a threshold policy.  At any time instant $m \in \{0,1,\cdots,M-1\}$, the optimal action is to stop exploration on current relay link and start exploring other relay links from $\mathbb{U}^i$ if $r_{m}\le \alpha_m$, or stop exploration on current relay link and choose it for data transmission if $r_{m}\ge \beta_m$, otherwise continue sending probe packet once more on the current relay link which is being explored to check the link quality if $\alpha_m < r_{m} < \beta_m$. Here $\alpha_m$, $\beta_m$ are appropriate constants.
		Also, we can say for the thresholds: $\cdots \ge \alpha_m \ge \alpha_{m-1} \ge  \cdots \ge \alpha_{1} $ and similarly  $\cdots \le \beta_m \le \beta_{m-1} \le  \cdots \le \beta_1$.
	
	\end{thm}
	\begin{proof}
		Proof given in Appendix \ref{A8}.	
	\end{proof}
	
	 Using above result, we can say as $m \rightarrow \infty$, $\alpha_m$ and $\beta_m$ converges to some scalar $\overline{\alpha}$ and $\overline{\beta}$ respectively. This is because $\beta_m$ is bounded below and $\alpha_m$ is bounded above. Hence, for very large horizon length $M$, the optimal policy can be approximated by a stationary threshold policy with a time-invariant threshold $\overline{\alpha}$ and $\overline{\beta}$. In this case we can say that at time instant $m$, if  $r_{m}\le \overline{\alpha}$  then stop exploration on current link and start exploring other links from $\mathbb{U}^i$, if $r_{m}\ge \overline{\beta}$ then stop exploration on current link and choose it for data transmission, otherwise  if $\overline{\alpha} < r_{m} < \overline{\beta}$ then continue sending probe packet once more on the current link which is being explored to check the link quality.

	It is easy to check that $\Phi(r,w)$ is a 
	non-decreasing function in $r$. Let us denote $r_0$ as the prior belief. We can say that, when prior belief satisfies $r_0>\Phi(r_0,w)=r_1$, then  $r_2=\Phi(r_1,0) < \Phi(r_0,0)=r_1$ and $r_2=\Phi(r_1,1) < \Phi(r_0,1)=r_1$. Proceeding in this way, we can show that $r_m$ strictly decreases with $m$ whenever we observe several successive ACK failures/successes. Similarly, when the prior belief satisfies $r_0<\Phi(r_0,w)=r_1$, then we can say that $r_2=\Phi(r_1,0) > \Phi(r_0,0)=r_1$ and $r_2=\Phi(r_1,1) > \Phi(r_0,1)=r_1$. Proceeding in this way, we can show that $r_m$ strictly increases with $m$ whenever we observe several successive ACK failures/successes. For getting $x$ successive ACK failures, we can define recursively a probability $\pi_x$ as: $\pi_1=\Phi(r_0,\overline{A})$,  $\pi_2=\Phi(\pi_1,\overline{A})$, $\cdots$, $\pi_x=\Phi(\pi_{x-1},\overline{A})$ with $\pi_0=r_0$.
	Similarly, for getting $x$ successive ACK successes, we can define recursively a probability $\pi_x'$ as: $\pi_1'=\Phi(r_0,A)$,  $\pi_2'=\Phi(\pi_1',A)$, $\cdots$, $\pi_x'=\Phi(\pi_{x-1}',A)$ with $\pi'_0=r_0$. Let $c$ and $d$ be the smallest integer such that $\pi_c \le \overline{\alpha}$ and $\pi_{d}' \ge \overline{\beta}$ respectively.  We can further simplify the stationary threshold policy as follows.
	\begin{cor}
		Using theorem \ref{thm_2}, we can simplify the optimal policy further as follows. Let $c$ and $d$ be the smallest integer such that $\pi_c \le \overline{\alpha}$ and $\pi_d' \ge \overline{\beta}$ respectively. If there are $c$ successive ACK failures, stop exploring on current relay link being explored and start exploring some other relay link. If there are $d$ successive ACK successes, stop exploring on current relay link being explored and choose it for data transmission. Otherwise, decision cannot be made and more probe packets are needed to be sent to learn the link quality.
	\end{cor}
			
	\section{Experiments and Results} \label{experiment}	
	\begin{figure*}[h!]
		\centering
		\begin{minipage}[b]{.22\textwidth}
			\includegraphics[width=1\textwidth]{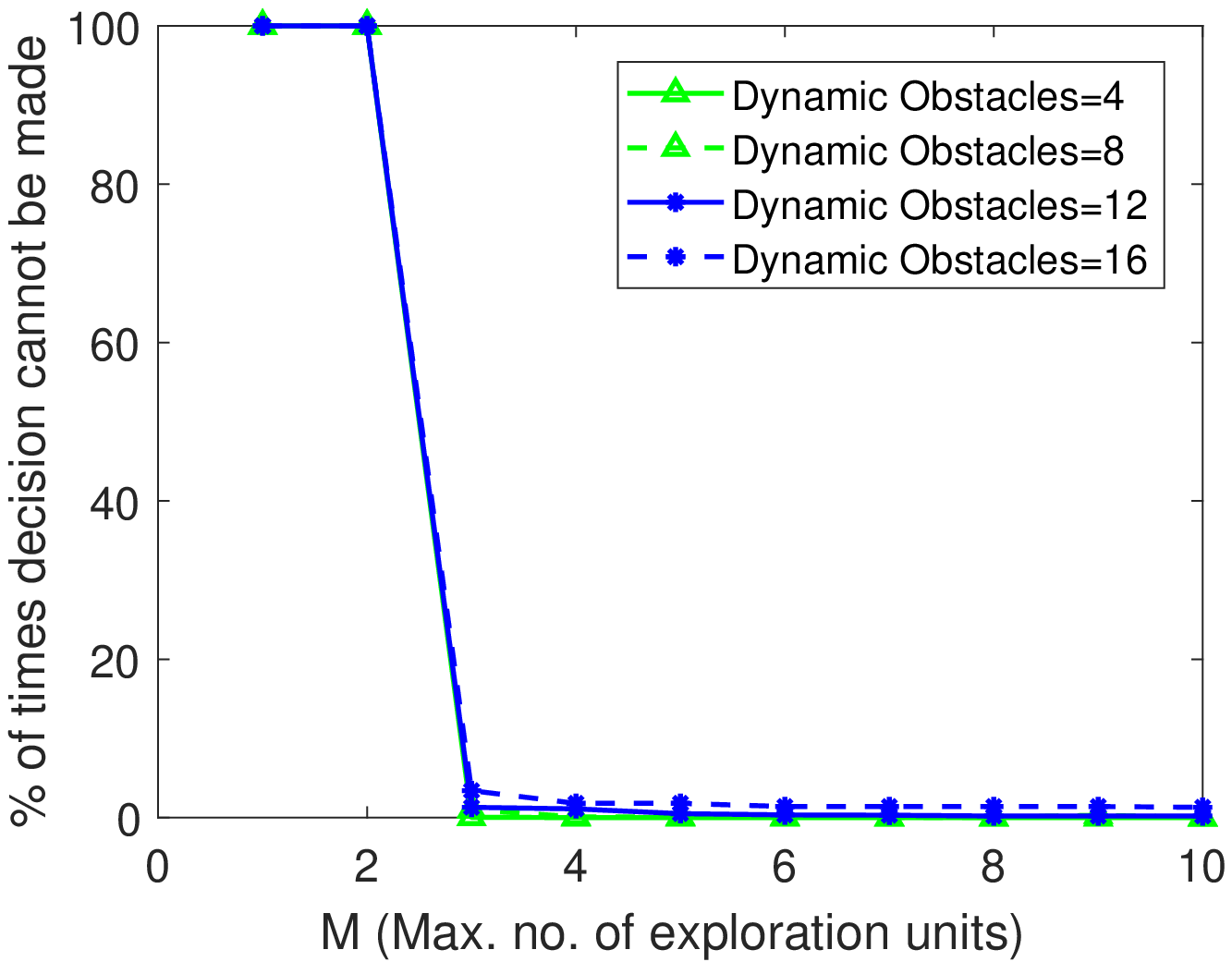}
			\caption{$M$ vs \% of times no decision can be made}
			\label{fig-1}
		\end{minipage}\qquad
	\begin{minipage}[b]{.21\textwidth}
		\includegraphics[width=1\textwidth]{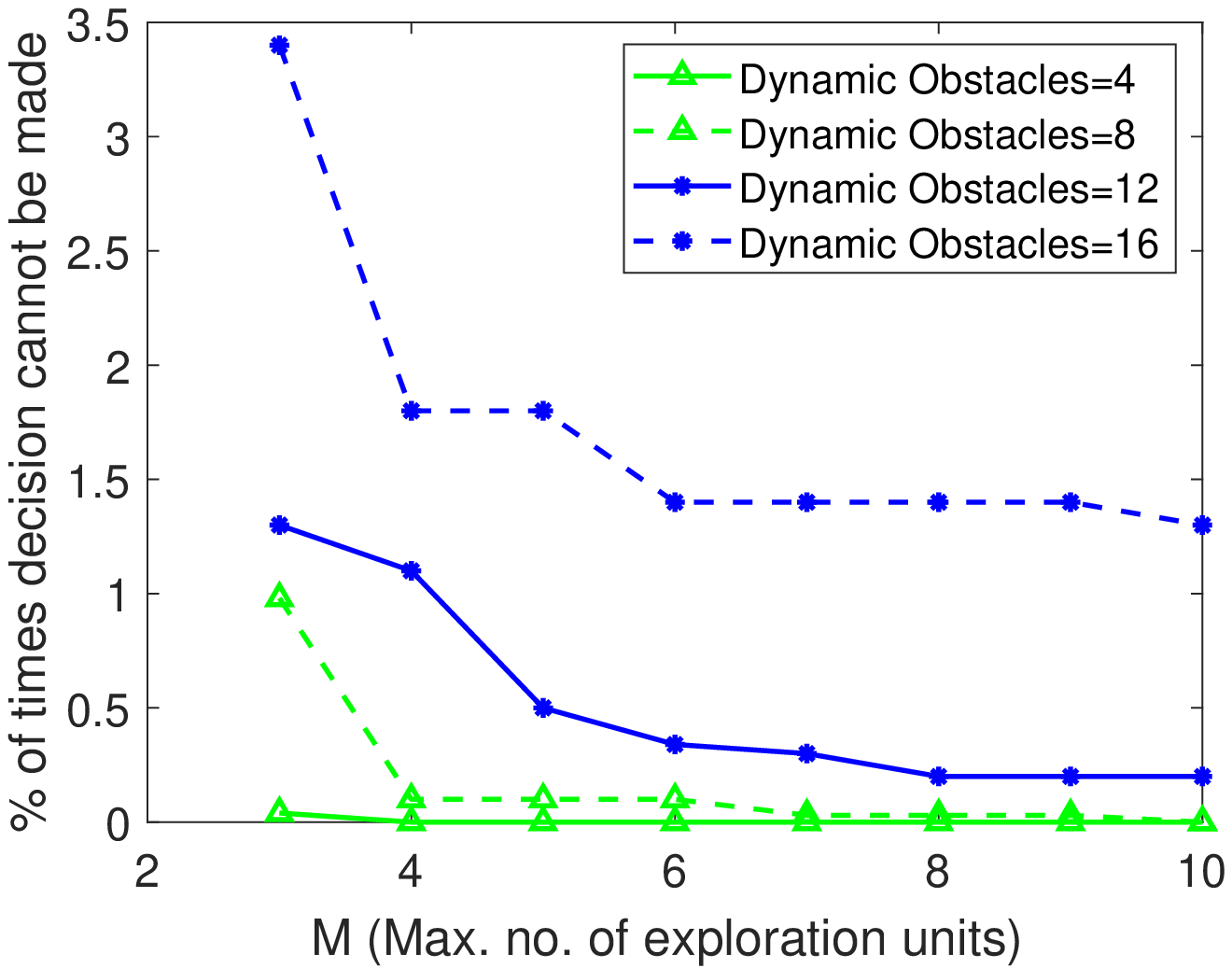}
		\caption{$M$ vs \% of times no decision can be made when $M\ge3$}
		\label{fig-1_2}
	\end{minipage}\qquad
		\begin{minipage}[b]{.21\textwidth}
			\includegraphics[width=1\textwidth]{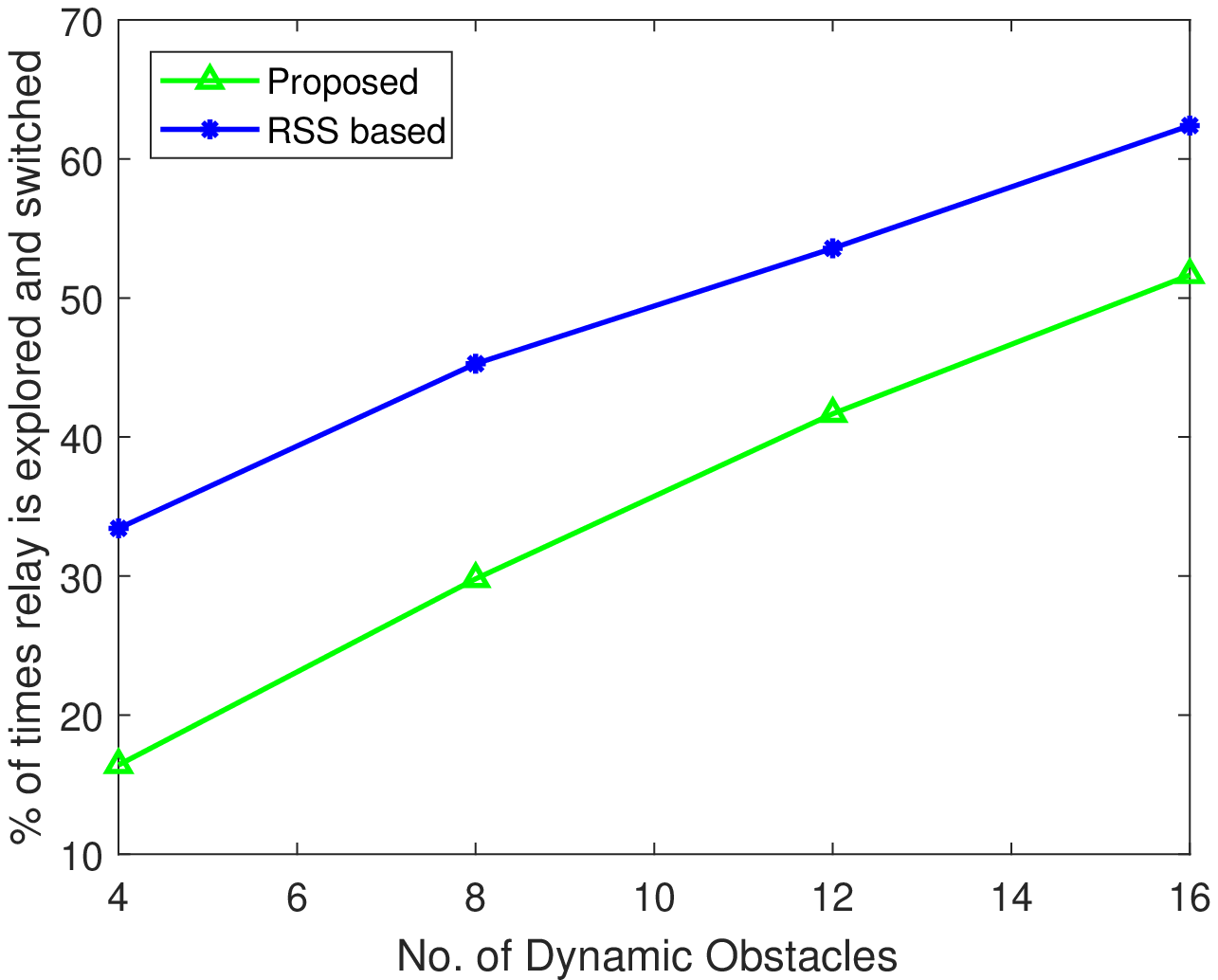}
			\caption{No. of times relay are explored and switched with dynamic obstacles.}
			\label{fig-2}
		\end{minipage}\qquad
		\begin{minipage}[b]{.24\textwidth}
			\includegraphics[width=1\textwidth]{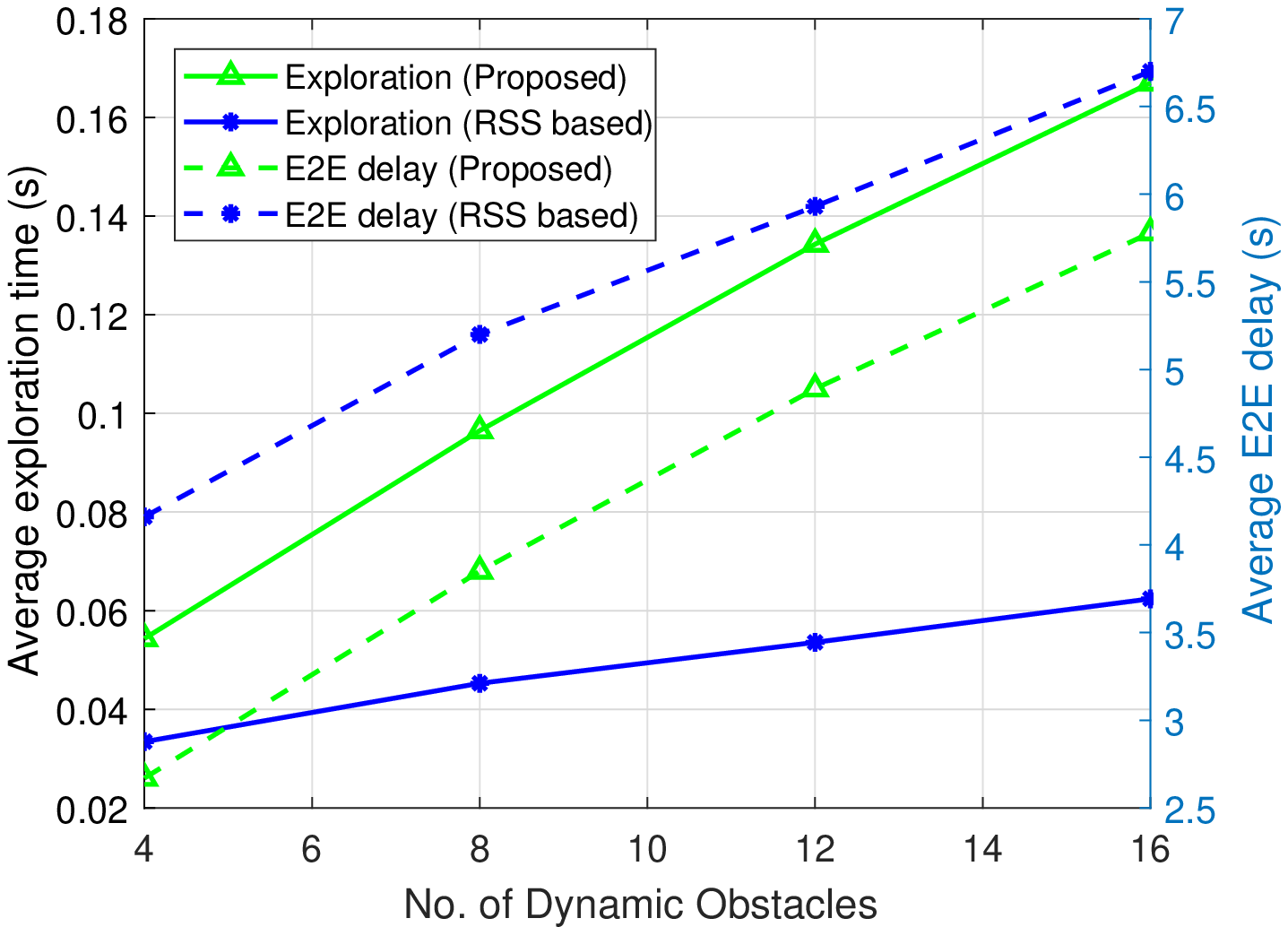}
			\caption{Trade-off for average exploration time and average E2E delay.}
			\label{fig-4}
		\end{minipage}
	\end{figure*}
	
	Service region of dimension $100~m \times 100~m$ is divided into grid zones each of dimension $10~m\times10~m$. We have taken $\delta=100~ms$ and $\epsilon=1~ms$.
	UEs are using directional transmitter and receiver antennas for $60~GHz$ frequency with $G_r=G_t=6~dB$. LOS path loss exponent is $2.5$ and zero mean  log-normal shadowing random variable with standard deviation $3.5$ \cite{7974772,7109864}. Thermal noise density is $-174~dBm/Hz$ and devices are using $24~dBm$ transmit power.  Capacity of each link $(i,j)$ is $B\log_2(1+S_{ij})~bits/sec$, where $B=20~MHz$ \cite{al2014path} is bandwidth and $S_{ij}$ is the received signal to noise ratio. The packet length is fixed and of size $65535~bytes$. A source UE in a grid zone $i$ can transmit data to a maximum of $24$ grids surrounding it ($|\mathbb{U}^i| \le 24$).  There are maximum  of $16$ static and $D$ dynamic obstacles present in these $24$ grids surrounding the source UE, where $D\in\{0,4,8,12,16\}$. Each static obstacle is placed uniformly in the service region. Each dynamic obstacles is moving randomly and independently of each other and following a simple blockage model such that with probability $0.5$ it will block a given link otherwise it will not block the link.  A single source-destination pair is assumed for simplicity and all other devices in a given zone may act as relay. We have written our own C++ custom code and run them on a \texttt{GNU} $4.8$ compiler on Intel core $i7$ machine and average of 1000 runs are taken. The average exploration time and E2E delay for sending $100$ packets are the main parameters considered. We have compared our results with the received signal strength (\texttt{RSS Based}) approach which is most commonly used for relay selection during exploration after transmitter and receiver beams are perfectly aligned.
	
	Value of $M$ is derived experimentally in figure \ref{fig-1}. This value should not be too small, otherwise there won't be sufficient exploration to lead to a decision. Also $M$ should not be too big, otherwise a lot of time would be wasted in exploration. Figure \ref{fig-1} shows the percentage of time the decision could not be made on varying  values of $M$ for different $D$. It is clear that for almost all values of $D$, decision cannot be made when $M$ is $1$ and $2$. For $D=16$ and $M=3$, for around $3.5\%$ of cases decision cannot be made a shown in figure \ref{fig-1_2}. This figure shows that with dynamic obstacles, no. of times the decision of relay selection cannot be made, also increases. We have chosen $M=4$ for performing further simulations.  
	
	Figure \ref{fig-2} shows the effect of dynamic obstacles on percentage of times new relays are explored and switched. With dynamic obstacles, the probability of link breakage increases and thus number of times new relays are explored and switched also increases. Our approach outperforms the RSS based approach because it  chooses the link by learning dynamic obstacle's presence while RSS based approach greedily chooses a link based on best RSS values which might have higher chance of blockage.	Figure \ref{fig-4} shows following results: first with dynamic obstacles, average E2E delay is increased which follows from previous result of figure \ref{fig-2}. Since higher number of dynamic obstacles will have higher chance of breaking a link and thus causes increase in E2E delay. Similarly, following the argument of previous result of figure \ref{fig-2}, our approach outperforms RSS based approach. Secondly, with dynamic obstacles, average exploration time is increased because higher link breakage with dynamic obstacles would cause more exploration and switching of relays. Here RSS based approach performs better since it always takes one time unit of exploration, whereas our approach takes additional exploration time units to learn the quality of link considering dynamic obstacles. Thirdly, it can be seen from above two results that, average exploration time is higher in our approach at the cost of reducing the cost of average E2E delay significantly compared to the RSS based approach. This describes the trade-off of additional exploration time  over the average E2E delay.

	\section{Conclusion}  \label{conclusion}
	The problem of selecting a given relay during exploration is investigated taking into account presence of dynamic obstacles. We have modeled this problem as a finite horizon POMDP	framework at each UE. Using this model, an optimal threshold policy is derived for each UE which is then simplified to a  stationary policy. This policy governs the source UE to  take decisions based on successive	ACK failures or success on the current relay link under exploration. Through simulations, the trade-off between average exploration time and E2E delay is shown, where our approach captures the effects of dynamic obstacles significantly compared to the RSS based approach.
	
	\bibliography{ref} 
	\bibliographystyle{ieeetr}

	\newpage
	
	\appendices

	\section{Proof of proposition \ref{prop_3}}
	\label{A6}
	
    We will prove this by first showing that $K_m(r)$ is piece-wise linear and concave for each $m$ using induction. Then we prove our proposition.
	For time instant $(M-1)$, we have, \[K_{M-1}(r)=\min\{rD_1,(1-r)D_2\}\] which is piece-wise linear and concave. 
	For time instant $(M-2)$, we have,
	\[K_{M-2}(r)=\min\{rD_1,(1-r)D_2,E_{M-2}(r)\}\] where,
	\begin{multline*}
		E_{M-2}(r)=c_\epsilon + \min\big\{\big(q r+ s(1-r)\big)kD_1,\\\big(r(1-q)+ (1-s)(1-r)\big)gD_2\big\}+ \\\min\big\{\big(q r+ s(1-r)\big)\big(1-k\big)D_1,\\\big(r(1-q)+ (1-s)(1-r)\big)\big(1-g\big)D_2\}.
	\end{multline*}
	 This overall equation is also piece-wise linear and concave.  
	
	Assuming $K_{m+1}(r)$ is piece-wise linear and concave in $r$, we can say that for some suitable scalars, $\eta_1,\eta_2, \cdots, \eta_n$ and $\beta_1,\beta_2, \cdots, \beta_n$, $K_{m+1}(r)$ can be written as:
	\begin{equation}\label{eq_l+1_linear_2}
	K_{m+1}(r)=\min\{\eta_1+\beta_1r,\eta_2+\beta_2r,\cdots,\eta_n+\beta_nr\}.
	\end{equation}	
	We can write,
	\[K_{m}(r)=\min\{rD_1,(1-r)D_2,E_m(r)\}.\] Expanding above using equation \eqref{eq_j15}, we get:
	\scriptsize
	\begin{multline}\label{eq_J_l_2}
	K_{m}(r)=\min\bigg\{rD_1,\big(1-r\big)D_2,c_\epsilon+\\
	\bigg(\big(rq+(1-r)s\big)k+\big(r(1-q)+(1-r)(1-s)\big)g\bigg)\times \\K_{m+1}\bigg(\frac{\big(rq+(1-r)s\big)k}{\big(rq+(1-r)s\big)k+\big(r(1-q)+(1-r)(1-s)\big)g}\bigg) +\\ \bigg(1-\big((rq+(1-r)s)k+(r(1-q)+(1-r)(1-s))g\big)\bigg)\times \\K_{m+1}\bigg(\frac{\big(rq+(1-r)s\big)\big(1-k\big)}{\big(1-((rq+(1-r)s)k+(r(1-q)+(1-r)(1-s))g)\big)}\bigg)\bigg\}.
	\end{multline}
	\normalsize
	Let us substitute $Y$ for $((rq+(1-r)s)k+(r(1-q)+(1-r)(1-s))g)$ to simplify the calculations. Now using equation \eqref{eq_l+1_linear_2}, the equation \eqref{eq_J_l_2} is reduced as:
	\begin{multline}\label{eq_J_l_expand_2}
	K_{m}(r)$=$\min\bigg\{rD_1,(1-r)D_2,c_\epsilon+\\
	Y\min\big\{\eta_1+\beta_1\frac{(qr+ s(1-r))k}{Y},\\ \eta_2+\beta_2\frac{(qr+ s(1-r))k}{Y},\cdots,\eta_n+ \beta_n\frac{(qr+ s(1-r))k}{Y}\big\}+\\(1-Y)\min\big\{\eta_1+\beta_1\frac{(qr+ s(1-r)) (1-k)}{(1-Y)}, \eta_2+\\ \beta_2\frac{(qr+ s(1-r)) (1-k)}{(1-Y)},\cdots,\\ \eta_n+ \beta_n\frac{(qr+ s(1-r)) (1-k)}{(1-Y)}\big\}\bigg\}
	\end{multline}
	We can further reduce above equation as:
	\begin{multline}\label{eq_J_l_linear_2}
	K_{m}(r)$=$\min\big\{rD_1,(1-r)D_2,c_\epsilon+
	\min\{(Y)\eta_1+\\ \beta_1(qr+ s(1-r))k, (Y)\eta_2+ \beta_2(qr+ s(1-r))k,\cdots,\\(Y)\eta_n+\beta_n(qr+ s(1-r))k\}+ \min\{(1-Y)\eta_1+\\ \beta_1(qr+ s(1-r)) (1-k),(1-Y)\eta_2+ \beta_2(qr+ s(1-r)) (1-k),\\\cdots,(1-Y)\eta_n+ \beta_n(qr+ s(1-r)) (1-k)\}\big\}
	\end{multline}
	This is again piece-wise linear and concave in $r$. Thus the induction is complete. 
	
	Now we will show that $E_m(r)$ is also piece-wise linear and concave in $r$:
	\begin{multline} \label{eq_A_l_proof_2}
	E_m(r)=c_\epsilon+(Y)K_{m+1}(\Phi'(r,1))+
	(1-Y)K_{m+1}(\Phi'(r,0))
	\end{multline}
	The first term $c_\epsilon$ is constant. For the next two terms $(Y)K_{m+1}(\Phi'(r,1))$ and $
	(1-Y)K_{m+1}(\Phi'(r,0))$, by expanding them using equation \eqref{eq_l+1_linear_2}, we get,
	\begin{multline}
	(Y)K_{m+1}(\Phi'(r,1))+
	(1-Y)K_{m+1}(\Phi'(r,0))=\\
	(Y)\min\bigg\{\eta_1+\beta_1\frac{(qr+ s(1-r))k}{Y}, \eta_2+\beta_2\frac{(qr+ s(1-r))k}{Y},\\ \cdots,\eta_n+ \beta_n\frac{(qr+ s (1-r))k}{Y}\bigg\}+	 (1-Y)\min\bigg\{\eta_1+\\ \beta_1\frac{(qr+ s(1-r))(1-k)}{1-Y}, \eta_2+\beta_2\frac{(qr+ s (1-r))(1-k)}{1-Y},\\ \cdots,\eta_n+ \beta_n\frac{(q r+s(1-r))(1-k)}{1-Y}\bigg\}
	\end{multline}
	We can reduce above to:
	\begin{multline}
		(Y)K_{m+1}(\Phi'(r,1))+
		(1-Y)K_{m+1}(\Phi'(r,0))=\\
		\min\big\{\eta_1(Y)+ \beta_1(qr+ s(1-r))k, \eta_2(Y)+\beta_2(qr+ s(1-r))k,\cdots,\\ \eta_n(Y)+\beta_n(qr+ s(1-r)) (1-k)\big\}+		 \min\big\{\eta_1(1-Y)+ \\\beta_1(qr+ s(1-r))(1-k), \eta_2(1-Y)+\beta_2(qr+ s(1-r)) (1-k),\cdots,\\ \eta_n(1-Y)+\beta_n(qr+ s(1-r)) (1-k)\big\}.
	\end{multline}
	Since minimum of finite number of concave function is concave, $E_m(r)$ is piece-wise linear and concave in $r$.
	
	\section{Proof of proposition \ref{prop_4}}
	\label{A7}
	
	
	We will first prove  $K_m(r)\ge K_{m+1}(r)$, then we will use this to prove $E_m(r)\ge E_{m+1}(r)$. First we start for base case $m=M-1$ and the first term in recursion $m=M-2$: \[K_{M-1}(r)=\min\{rD_1,(1-r)D_2\}\] and \[K_{M-2}(r)=\min\{rD_1,(1-r) D_2,E_{M-2}(r)\}\] respectively. We can easily see that $K_{M-1}(r)\ge K_{M-2}(r)$. We now prove it for first two terms of the recursion  $K_{M-2}(r)$ and $K_{M-3}(r)$. Let's denote $Y=(rq+(1-r)s)k+(r(1-q)+(1-r)(1-s))g$, we can write $K_{M-3}(r)$ as:
	\begin{align*}
	K_{M-3}(r)&= \min \{rD_1,(1-r)D_2, c_\epsilon +  \notag 
	(Y)K_{M-2}(\Phi'(r,1))
	+ \\ &\notag 
	\quad\quad
	(1-Y)K_{M-2}(\Phi'(r,0))\} \\
	&\le \min\{rD_1,(1-r)D_2, c_\epsilon +   \notag 
	(Y)K_{M-1}(\Phi'(r,1))+ \\ &\notag 
	\quad\quad
	(1-Y)K_{M-1}(\Phi'(r,0))\} \\
	&=K_{M-2}(r)
	\end{align*}
	Hence $K_{M-2}(r)\ge K_{M-3}(r)$. Similarly it proceeds for other $m$ and hence $K_{m+1}(r)\ge K_{m}(r)$. Now let us see this for $E_m(r)$ using previous proof for $K_{m}(r)$:
	\begin{align}
	E_{m}(r)&= c_\epsilon +
	(Y)K_{m+1}(\Phi'(r,0))
	 +(1-Y)K_{m+1}(\Phi'(r,0))\\
	&\le c_\epsilon +
	(Y)K_{m+2}(\Phi'(r,0))+
	 ((1-Y)K_{m+2}(\Phi'(r,0)) \\
	&=E_{m+1}(r)
	\end{align}
	Hence  $E_{m+1}(r) \ge E_{m}(r)$.

	\section{Proof of theorem \ref{thm_2}}
	\label{A8}
	
	For the last time period $M-1$, there exists $\rho=\frac{D_2}{D_1+D_2}$ such that stop exploration on current link and start exploring other links from $\mathbb{U}^i$ if $r<\rho$, otherwise stop exploration on current link and choose it for data transmission if $r\ge \rho$. At $\rho$, $K_{M-1}(\rho)$ attains it maximum value of $\frac{D_2D_1}{D_1+D_2}$.
	
	For general time instants $m$, we can say that $E_m(0)>c_{\epsilon}$ and $E_m(1)>c_{\epsilon}$. Using this fact, proposition \ref{prop_3} and proposition \ref{prop_4}, we can say that  if for some $r'$, $E_{M-2}(r') < \frac{D_2D_1}{D_1+D_2}$, then $E_{M-2}(r)$ and subsequent $E_m(r)$ will intersect $\min\{rD_1,(1-r)D_2\}$ at two points. Hence we can say that for general time instant $m$ we will get an optimal policy as: stop exploration on current link and start exploring other links from $\mathbb{U}^i$ if $r_{m}\le \alpha_m$, or stop exploration on current link and choose it for data transmission if $r_{m}\ge \beta_m$, otherwise continue sending probe packet once more on the current link which is being explored to check the link quality if $\alpha_m < r_{m} < \beta_m$. Here $\alpha_m$ and $\beta_m$ are found by satisfying the corresponding equations: $\alpha_m D_1=E_m(\alpha_m)$ and $(1-\beta_m)D_2=E_m(\beta_m)$. 
	
	If  for all $r'$, $E_{M-2}(r') > \frac{D_2D_1}{D_1+D_2}$, then terms $rD_1$ and $(1-r)D_2$ will contribute to the single threshold $\rho$ which is trivial for time instant $M-2$ due to proposition \ref{prop_3}. In this case, using  proposition \ref{prop_3} and proposition \ref{prop_4} one can have either single or two thresholds for other time instants $m$ depending upon whether $E_m(r)$ intersects $\min\{rD_1,(1-r)D_2\}$ at two or no points as shown in figure below.
	
	\begin{figure}[h!]
		\centering
		\includegraphics[width=0.5\textwidth]{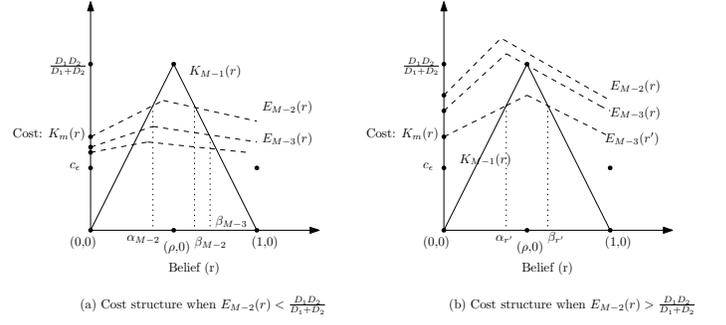}
		\caption{Cost structure of the problem for two cases.}
		\label{fig_proof}
	\end{figure}
	
	For the second part, using proposition \ref{prop_4} and if $E_{M-2}(\rho) < \frac{D_2D_1}{D_1+D_2}$,  we can say that $\alpha_{M-2} < \rho < \beta_{M-2}$. This is easy to see because $K_{M-1}(r)$ first increases till $\rho$ and then decreases. Also,
	\[K_{M-1}(r) \ge K_{M-2}(r)=\min\{K_{M-1}(r),E_{M-2}(r)\}\]
	and  $\alpha_{M-1}=\beta_{M-1}=\rho$. Hence, $\alpha_{M-2} \le \alpha_{M-1}$ and $\beta_{M-1} \le \beta_{M-2}$.
	Similarly, using proposition \ref{prop_4} we can say this for other instants $m$  that $\alpha_{m} \le \alpha_{m+1}$ and $\beta_{m} \ge \beta_{m+1}$. Hence we can say that with respect to $m$:  \[\cdots \ge \alpha_m \ge \alpha_{m-1} \ge  \cdots \ge \alpha_{1} \] and similarly  \[\cdots \le \beta_m \le \beta_{m-1} \le  \cdots \le \beta_1.\]


\end{document}